\documentclass[11pt]{article}
\usepackage{authblk}
\title{A Quantum Energy Inequality for a Non-commutative QFT} 
\author[1]{Harald Grosse\footnote{harald.grosse@univie.ac.at}} 
\author[2]{Albert Much\footnote{much@itp.uni-leipzig.de}}  
\affil[1]{Fakult\"at f\"ur Physik, Universit\"at Wien, Boltzmanngasse 5, A-1090 Wien, Austria}
\affil[2]{Institut f\"ur Theoretische Physik, Br\"uderstra\ss e 16\\ Universit\"at Leipzig\\ D-04103 Leipzig, Germany} 
\usepackage[a4paper, margin=25mm, bottom=40mm]{geometry}
\usepackage{setspace}
\usepackage{xcolor}
\usepackage[linktocpage=true]{hyperref}
\hypersetup{colorlinks=true, linkcolor=teal, urlcolor=teal, citecolor=teal}

\usepackage{titlesec}
\titleformat*{\section}{\Large \bfseries}
\titleformat*{\subsection}{\large \bfseries}

\usepackage{amsthm}    
\usepackage{amsmath}
\usepackage{amssymb}
\usepackage{physics}
\usepackage{bbm}
\usepackage{bm}
\usepackage{mathrsfs}

\newcommand{\bb}[1]{\mathbbm{#1}}
\newcommand{\e}{\mathrm{e}}
\renewcommand{\i}{\mathrm{i}}
\newcommand{\ord}[1]{\,:\,\mathrel{#1}\,:\,}

\newcounter{stm}
\counterwithin{stm}{section}
\numberwithin{equation}{section}

\newtheorem{theorem}{\textsc{Theorem}}[section]
\newtheorem{lemma}[theorem]{\textsc{Lemma}}
\newtheorem{proposition}[theorem]{\textsc{Proposition}}

\newcommand{\R}{\mathbb{R}}

\begin{document}
\maketitle
\begin{abstract}
We present a  quantum energy inequality (QEI) for quantum field theories formulated in non-commutative spacetimes, extending fundamental energy constraints to this generalized geometric framework. By leveraging operator-theoretic methods inspired by the positivity map of Waldmann et al. \cite{waldmannpos}, we construct  linear combinations of deformed operators that generalize the commutative spacetime techniques of Fewster  et al., \cite{Few98}. These non-commutative analogs enable us the derivation of a lower bound on the deformed averaged energy density, ensuring the stability of the underlying quantum field theory.  Our result establishes rigorous constraints on the expectation values of the deformed (non-commutative) energy density, reinforcing the physical consistency of non-commutative models while preserving core principles of quantum field theory.
\end{abstract} 
 \section{Introduction}
Quantum field theory (QFT) is traditionally formulated under the assumption that spacetime is a continuous and commutative manifold, allowing for the arbitrary localization of quantum fields. However, various approaches to quantum gravity suggest that at extremely short distances, spacetime may acquire a non-commutative structure, fundamentally altering key concepts such as locality, causality, and quantum fluctuations. One way to model such effects is by introducing   coordinate operators $q$ that satisfy the commutation relations:
\begin{align}  \label{eqnccoord}
    [q_{\mu},q_{\nu}] = i\Theta_{\mu\nu}  ,
\end{align}  
where $\Theta$ is a skew-symmetric matrix that characterizes the scale of non-commutativity. This structure implies an inherent uncertainty in position measurements, which prevents the exact localization of fields and events. As a result, quantum fluctuations—fundamental to QFT—are modified in a non-commutative setting.

These modifications have profound implications for physical quantities, such as the energy density. In standard QFT, the energy density of a quantum field is not necessarily positive-definite and can assume negative values due to quantum fluctuations. In non-commutative spacetime, where fluctuations are altered, understanding and constraining such negative energy densities becomes even more critical. This brings us to quantum energy inequalities (QEIs), which play a crucial role in ensuring that these negative energy densities remain physically constrained.

 QEIs are essential in standard QFT as well, where negative energy densities are linked to exotic physical effects, including violations of the weak energy condition and potential applications in wormhole physics and the Casimir effect. Without proper bounds, unrestricted negative energy densities could lead to pathological behaviors such as instabilities or violations of causality (see \cite{IN8, Kontou_2020} and references therein).  To address these concerns, QEIs provide bounds on the expectation values of the energy density, ensuring that negative energy densities, while permissible, remain physically constrained, see \cite{NCQEI1, NCQEI2, NCQEI3, NCQEI4, NCQEI5,  Fewster_1999, Fewster_2000, Fewster_2002, Fewster_2003, Fewster_2003a, Fewster_2005, Fewster_2006, Fewster_2006a, Fewster_2006b, Fewster_2007, Fewster_2007b, Fewster_2008, Fewster_2009, Boss, fröb2023quantumenergyinequalitysinegordon}.

Quantum energy inequalities play a vital role in constraining spacetime geometries that permit causality-violating solutions. This importance, beyond their function in restricting negative energy, also extends  their relevance in the context of non-commutative quantum field theory (NCQFT).  This theory is considered a leading candidate for quantum gravity that inherently modifies the spacetime structure at fundamental scales. While NCQFT's nonlocal interactions could, in principle, enable acausal phenomena, this very prospect reinforces the necessity of generalizing QEIs to non-commutative geometries. Such inequalities would act as safeguard conditions, ensuring that causality constraints persist even in quantum regimes where traditional concepts of spacetime break down. This work aims to explore these generalizations and their implications.

In this paper, we present a quantum energy inequality for a quantum field theory formulated in non-commutative Minkowski spacetime. In~\cite{DFR}, a QFT in non-commutative Minkowski space was developed by defining its representation space as the tensor product $\mathcal{V} \otimes \mathscr{F}_{s}(\mathscr{H})$, where $\mathscr{H}$ is the one-particle Hilbert space and $\mathscr{F}_{s}(\mathscr{H})$ the corresponding symmetric Fockspace of the Klein-Gordon quantum field $\phi$ and and  $\mathcal{V}$ is a representation space of the 
   non-commutative coordinates $q$ (see Equation \eqref{eqnccoord}). Later, in Ref.~\cite{GL1}, a unitary operator was constructed to map the tensor product space $\mathcal{V} \otimes \mathscr{F}_{s}(\mathscr{H})$ onto $\mathscr{F}_{s}(\mathscr{H})$  under which the field takes the form  
\begin{align*}
 \phi_{\Theta}(f) &=   \int d^4x \,f(x)\,\phi_{\Theta}(x) \\ 
 &= \int \frac{d^3\mathbf{k}}{2\omega_{\mathbf{k}}}\,\left(  f^{-}(\bm{k})\, e^{ \frac{i}{2}k\Theta P} a (\bm{k})+ f^{+}(\bm{k})\,e^{-\frac{i}{2} k\Theta P}  a^*(\bm{k}) \right),
\end{align*}  
where    \( \omega_{\mathbf{k}} = \sqrt{\bm{k}^2+m^2} \) and \( k = (\omega_{\mathbf{k}},\bm{k}) \). Here, $f$ is a test function,   \( P \) denotes the quantum field theoretical (on-shell) momentum operator and $a$ and $a^*$ represent the particle annihilation and creation operators with definite momenta that fulfill the canonical commutation relations. Moreover,  the non-commutativity of the spacetime is encoded in the skew-symmetric  matrix of the form \begin{equation}\label{eqtheta1}
  \Theta_{\mu\nu}=  \left(\begin{matrix} 0 &\theta & 0&0\\ -\theta &0 & 0&0 \\0 &0 & 0&\theta'\\0 &0 & -\theta'&0\end{matrix}\right),
\end{equation}
where  $\theta, \theta'\in\mathbb{R}$. In order to simplify the calculations, we assume $\theta = \theta'$ in the following analysis.  These constants are from a physical perspective of Planck length order squared (see \cite{DFR}). The representation on $ \mathscr{F}_{s}(\mathscr{H})$ led to a strict deformation quantization known as \textit{warped convolutions}~\cite{BLS}. The wedge-locality proof for deformed fields in \cite{GL1} further requires the deformation parameter $\theta\in\R^{+}$ to satisfy a positive-definiteness condition. We explicitly adopt this assumption and maintain it consistently throughout the analysis.

The Grosse-Lechner framework for non-commutative quantum field theory in Minkowski spacetime, as introduced in \cite{GL1}, can be interpreted as an interacting model, specifically a field theory with a non-trivial S-matrix. This model belongs to the class of interacting theories characterized by factorizing S-matrices. Consequently, establishing a quantum energy inequality (QEI) for this particular model is tantamount to demonstrating a QEI for an interacting quantum field theory. This result has the potential to serve as a basis for extending QEIs to more general frameworks. In particular, a QEI has already been established for an interacting case, namely for the massive Ising model in \cite{Boss}.

The quantum energy inequality that we examine is inspired by the approach in \cite{Few98}, where the authors establish a quantum energy inequality for a scalar field in commutative Minkowski spacetime. In their work, they impose constraints on the negative energy densities that can arise in quantum field theory, providing essential bounds to ensure consistency with fundamental physical principles. Their approach involves weighted averages of the energy density over specific regions, demonstrating that negative energy densities cannot persist indefinitely or accumulate arbitrarily.

The proof of the quantum energy inequality for a non-commutative quantum fields proceeds as follows. We begin by defining the operator \( X^{\pm}_{\omega} \):  
\begin{align} \boxed{
    X^{\pm}_{\omega} = \int \dd[3]{\bm{k}} \, p(\bm{k}) \left( g(\omega-\omega_{\bm{k}}) \,e^{ \frac{i}{2}k\Theta P} a (\bm{k})  \pm
    g(\omega+\omega_{\bm{k}})\, e^{- \frac{i}{2}k\Theta P}  a^*(\bm{k}) \right) }
\end{align}  
with a function $g$ that is smooth, real-valued and even, decays rapidly at infinity and where $p(\bm{k})$ is  a real valued function on $\mathbb{R}$ growing no faster than polynomial.
Next, we consider the deformed product of this operator with its adjoint,  $X^{\pm*}_{\omega} \times_{\theta} X^{\pm}_{\omega}$. In general, this expression is not    a positive operator. However, by applying the Waldmann positivity map \cite{waldmannpos}, denoted by \( S_{\theta} \),     the deformed product is mapped to a positive element of the algebra, and by integrating this expression over the variable \( \omega \) we obtain  a positive operator:  
\begin{align*} 
   \int_{0}^{\infty} \dd{\omega} \, S_{\theta} ( X^{\pm *}_{\omega} \times_{\theta} X^{\pm}_{\omega}) \geq 0. 
\end{align*}  
Positivity holds as well for the Waldmann map applied to the undeformed  product: 
\begin{align*} 
   \int_{0}^{\infty} \dd{\omega} \, S_{\theta} (X^{\pm*}_{\omega}   X^{\pm}_{\omega}) \geq 0. 
\end{align*}  
By forming  a linear combination of those two operators we obtain an inequality from which we extract a deformed operator along with an additional integral term, that is a $c$-number:  
\begin{align*}\boxed{ 
 \int_{0}^{\infty}  \frac{\dd{\omega}}{2}  \,\Bigr( S_{\theta} (X^{\pm*}_{\omega} \times_{\theta} X^{\pm}_{\omega}) + S_{\theta} (X^{\pm*}_{\omega}  X^{\pm}_{\omega}) \Bigl)
      = W_{\Theta}^{\pm} + \int_{0}^{\infty}\dd{\omega}\int\dd[3]{\bm{k}} \, 2\omega_{\bm{k}}\,p(\bm{k})^2  g(\omega+\omega_{\bm{k}})^2 \, \geq 0.}
\end{align*}  
Finally, identifying the operator \( W_{\Theta}^{\pm} \) with the smeared, normal ordered and  deformed energy density $\ord{{T_{00}^\Theta}}$ we complete the proof of the quantum energy inequality. The smeared, normal ordered and  deformed energy density on the other hand is obtained by replacing the fields on a commutative Minkowski spacetime with those quantum fields living in  a non-commutative spacetime, compare with  \cite[Equation (54)]{NCEMT2000}  and see as well \cite{NCEMT2001,  NCEMT2003, NCEMT2015, NCEMT2018} 
\begin{align}\boxed{
  \ord{{T_{00}^\Theta}}=\frac{1}{2}\ord{\bigl\{ 
  \left(
  \partial_0\phi_{\Theta}\right)^2+\sum_{i=1}^{3}(\partial_i\phi_{\Theta})^2+m^2(\phi_{\Theta})^{  2}
  \bigr\}}.}
\end{align}

 \section{The Non-commutative Energy Density}
 In this section we define, as in  \cite[Equation (54)]{NCEMT2000}\footnote{The authors call this the algebraic expression, which is equivalent to our expression if one identifies the fields on the tensor product space $\mathcal{V} \otimes \mathscr{F}_{s}(\mathscr{H})$ with the fields on $\mathscr{F}_{s}(\mathscr{H})$ as was done in \cite{GL1} and in \cite[Equation 1.3]{GL2}. } the deformed energy density by replacing the fields $\phi$ by their non-commutative counterparts $\phi_{\Theta}$. In order to derive a quantum energy inequality we first define the Wick-ordering prescription for the non-commutative creations and annihilation operators that are defined by
\begin{align}\label{eqdefaop}
a_{\Theta}(\bm{k}): =     e^{\frac{i}{2}k\Theta P} a (\bm{k}),
\qquad \qquad 
a^*_{\Theta}(\bm{k}): =
e^{-\frac{i}{2}k\Theta P}  a^*(\bm{k}) ,
\end{align} 
where $k\Theta k'=k_{\mu}\Theta^{\mu\nu}k'_{\nu}$ and the twisted CCR algebra   for arbitrary on-shell momenta\footnote{$\mathscr{H}^{+}_{m}$ denotes the upper mass shell.}   $k, k'\in	\mathscr{H}^{+}_{m}$ is given by
	\begin{align}\label{eqtccr}
	    a_{\Theta}(\bm{k})a_{\Theta}(\bm{k}')&=e^{-ik\Theta k'} a_{\Theta}(\bm{k}') a_{\Theta}(\bm{k})\\\label{eqtccr1}
 a^*_{\Theta}(\bm{k}) a^*_{\Theta}(\bm{k}')&=e^{-ik\Theta k'} a^*_{\Theta}(\bm{k}') a^*_{\Theta}(\bm{k})
 \\\label{eqtccr2}
 a_{\Theta}(\bm{k}) a^*_{\Theta}(\bm{k}')&= e^{ik\Theta k'}a^*_{\Theta }(\bm{k}') a_{\Theta}(\bm{k})+2\omega_{\mathbf{k}}\,\delta^3(\bm{k}-\bm{k}') .
 \end{align}
The Wick-ordering prescription is given by 
\begin{align*}
    \ord{\phi_{\Theta}(x)\phi_{\Theta}(y)}=\phi_{\Theta}(x)\phi_{\Theta}(y)-\langle \Omega\vert \phi_{\Theta}(x)\phi_{\Theta}(y)\Omega\rangle,
\end{align*}
 where $\vert\Omega\rangle$ is the vacuum-vector. See Appendix \ref{appnorm} for a   proof.  It is easily seen, due to the invariance of the vacuum under application of the momentum operators (and any function thereof), that the vacuum expectation value of the square of two deformed scalar fields is equal to the undeformed two-point function. Hence, we can write the Wick-ordering prescription as
 \begin{align}\boxed{
    \ord{\phi_{\Theta}(x)\phi_{\Theta}(y)}=\phi_{\Theta}(x)\phi_{\Theta}(y)-\langle \Omega\vert\phi(x)\phi(y) \Omega\rangle.} 
\end{align}
This translates  explicitly for the term $\ord{a_{\Theta}(\bm{k})a^*_{\Theta}(\bm{k}')}$  as follows (see Appendix \ref{appnorm}),
 \begin{align}\boxed{
   \ord{a_{\Theta}(\bm{k})a^*_{\Theta}(\bm{k}')}= e^{ik\Theta k'}a^*_{\Theta}(\bm{k}')a_{\Theta}(\bm{k}).}
\end{align}
Using this non-commutative Wick-ordering prescription the renormalized non-commutative energy density then reads
\begin{align} \label{eq:rdemt} 
	\ord{T_{f_{\theta}}^\Theta}&= \int_{-\infty}^\infty \dd{t}\int \dd[3]{\bm{x}} \ord{T_{00}^\Theta (t, \bm{x})} f_{\theta}(t,\bm{x}) \\
	&= \frac{1}{2} \iint \frac{\dd[3]{\bm{k}} \dd[3]{\bm{k}'}}{ 4 {\omega_{\bm{k}} \omega_{\bm{k}'}}} \Bigl[(\omega_{\bm{k}} \omega_{\bm{k}'} + \bm{k} \cdot \bm{k}') \nonumber \\
	&\hspace{40pt} \Bigl(\e^{\i k \Theta k'} a_\Theta^* (\bm{k}') a_\Theta (\bm{k}) \widehat{f_{\theta} }(\omega_{\bm{k}} - \omega_{\bm{k}'}, {\bm{k}} - {\bm{k}'}) + a_\Theta^* (\bm{k}) a_\Theta (\bm{k}') {\widehat{f_{\theta} }(\omega_{\bm{k}} - \omega_{\bm{k}'}, {\bm{k}} -  {\bm{k}'}) }\nonumber \\
	&\hspace{60pt} - a_\Theta (\bm{k}) a_\Theta (\bm{k}') \widehat{f_{\theta} }(\omega_{\bm{k}} + \omega_{\bm{k}'}, {\bm{k}} +  {\bm{k}'}) - a_\Theta^* (\bm{k}) a_\Theta^* (\bm{k}') {\widehat{f_{\theta} }(\omega_{\bm{k}} + \omega_{\bm{k}'},{\bm{k}} + {\bm{k}'})} \Bigr) \nonumber \\
	&\quad + m^2 \Bigl(\e^{\i k \Theta k'} a_\Theta^* (\bm{k}') a_\Theta (\bm{k}) \widehat{f_{\theta} }(\omega_{\bm{k}} - \omega_{\bm{k}'}, {\bm{k}} - {\bm{k}'})+ a_\Theta^* (\bm{k}) a_\Theta (\bm{k}') {\widehat{f_{\theta} }(\omega_{\bm{k}} - \omega_{\bm{k}'}, {\bm{k}} -  {\bm{k}'}) }\nonumber \\
	&\hspace{60pt} + a_\Theta (\bm{k}) a_\Theta (\bm{k}') \widehat{f_{\theta} }  (\omega_{\bm{k}} + \omega_{\bm{k}'}, {\bm{k}} +  {\bm{k}'}) + a_\Theta^* (\bm{k}) a_\Theta^* (\bm{k}') {\widehat{f_{\theta} }(\omega_{\bm{k}} + \omega_{\bm{k}'}, {\bm{k}} +  {\bm{k}'}) }\Bigr) \Bigr].  \nonumber 
\end{align}
where have smeared it with   a smooth, even, and non-negative function   $f_{\theta}(t,\bm{x})$ on $\bb{R}^4$ which decays rapidly at infinity (and is possibly compactly supported), and where the  Fourier transform $\widehat{f_{\theta}}$ of the function $f_{\theta}$ is given by the (real) function 
\begin{align}
    \widehat{f_{\theta}}(\omega_{\bm{k}},\bm{k}) =\int_{-\infty}^{\infty}dt\int d^3\bm{x}\,e^{-i\omega_{\bm{k}}t+i\bm{k}\bm{x}} f_{\theta}(t,\bm{x}).
\end{align}

Without loss of generality we  choose the test function $f_{\theta}(t,\bm{x})$ as  a product of two test functions
\begin{align}\label{eqffh}
    f_{\theta}(t,\bm{x})= f_{\theta}(t){h}_{\theta}( \bm{x})  ,
\end{align}
where we choose the spatial test function ${h}_{\theta}$ to be the delta function in the limit $\theta\rightarrow0$, i.e.\
\begin{align}\label{eqhtheta1}
  {h}_{\theta}( \bm{x}) =\frac{e^{-\frac{1}{\theta}\Vert\bm{x}\Vert^2}}{\sqrt{\theta}}.
\end{align}
\section{The Waldmann Positivity Map and a Positive Operator}
In our proof of the quantum energy inequality, we employ the Moyal star (or Rieffel) product to generate essential non-commutative terms in the operator algebra. However, the inherent limitation arises because the star product of an operator  and its adjoint fails to preserve positivity in the deformed algebra. To resolve this, we introduce a positive-restoring map $S_{\theta}$, which acts as a morphism between the star-product algebra and a space of operators equipped with a positive-definite   product structure. This mapping ensures the physical viability of energy density expectations while retaining critical non-commutative features. In Ref.~\cite{waldmannpos}, the authors construct explicitly a map that takes the star product of two elements in the (star) algebra $\mathcal{A}$, specifically $X$ and $X^*$ with $X^*, X\in \mathcal{A}$  and associates it with positive elements of the same algebra.  To ensure clarity and ease its application, we provide a concise summary of this construction below. For a more detailed exposition, we refer the reader to the original work. First, let $g:V\times V\to \mathbb{R}$ be a positive inner product on $V$. Then,  define the linear  operator $S_{\theta}:\mathcal{A}\to\mathcal{A}$ for $X\in \mathcal{A}$ as follows, 
\begin{equation}\boxed{
     S_{ {\theta} }(X)=\pi^{-2}\int_V \dd[4]{u}\,e^{-g(u,u)}\alpha_{\sqrt{\theta}u}(X).}
\end{equation}
Using the operator $S_{\theta}$ and assuming that $g$ is a compatible\footnote{Compatibility refers with regards to the symplectic structure $\omega$ which is the inverse of the Poisson tensor $\Theta$, i.e.\ $g(X,Y)=\omega(X,JY)$, where $J$ is a complex structure.} positive definite form, it follows by \cite[Theorem 3.3]{waldmannpos} that the deformed product of $X^*$ and $X$ given by
\begin{align} 
    X^*\times_{\theta}X= (2\pi)^{-4} \iint_{V\times V} \dd[4]{v} \dd[4]{z}\,\alpha_{\Theta v}(X^*)\alpha_{z}(X )\,e^{-i vz} , 
\end{align}
lies in the set of positive elements of the algebra, denoted by $\mathcal{A}^{+}$, i.e.\
\begin{align*}
      S_{ {\theta}}(X^*\times_{\theta}X)\in \mathcal{A}^{+}.
\end{align*} 
Next, we define the  operator 
\begin{align*}
  \boxed{ X^{\pm}_{\omega}= \int \dd[3]{\bm{k}} \,p(\bm{k})\left(  g(\omega-\omega_{\bm{k}})a_{\Theta} (\bm{k} )\pm
    g(\omega+\omega_{\bm{k}})a_{\Theta}^*(\bm{k} )\right)}
\end{align*}
 and we use the linear operator $S_{\theta}$, while choosing the positive definite metric to be $g=\delta$ the Kronecker delta and the vector space $V=\R^4$, to define the following manifestly positive operator 
    \begin{align}\label{eqs1}
    \int_{0}^{\infty}\dd{\omega} \,    S_{\theta}\left(    X^{\pm*}_{\omega} \times_{\theta}X^{\pm}_{\omega} 
     \right) .
\end{align}
\begin{lemma}\label{lemm1}
    The positive operator given in Equation \eqref{eqs1} is explicitly given by 
    \begin{align*}
    \int_{0}^{\infty}\dd{\omega} \,    S_{\theta}\left(     X^{\pm*}_{\omega} \times_{\theta}X^{\pm}_{\omega} 
     \right)    & =
 W^{1,\pm}_{\Theta}+   \int_{0}^{\infty}\dd{\omega} \int \dd[3]{\bm{k}} \, 2\omega_{\bm{k}}\, p(\bm{k})^2
g(\omega + \omega_{\bm{k}}) ^2  ,\, 
\end{align*}
where  the operator $W^{1,\pm}_{\Theta}$ is defined by
\begin{align}\label{eqspmW1}
     W^{1,\pm}_{\Theta}&:= \iint \dd[3]{\bm{k}}'\dd[3]{\bm{k}} \,  p(\bm{k})p(\bm{k}')\\
	& \nonumber
\quad \Bigl(\Bigl(e^{-\frac{\theta}{4}\Vert k'-k \Vert^2}\, \int_{0}^{\infty} \dd{\omega}\, g(\omega-\omega_{\bm{k}'}) g(\omega-\omega_{\bm{k}})\Bigr)\, e^{ik\Theta  k'}
  a^*_{\Theta} (\bm{k}' )a_{\Theta} (\bm{k}  )  \\&\nonumber
  \quad +\Bigl( e^{-\frac{\theta}{4}\Vert k'-k \Vert^2}\int_{0}^{\infty} \dd{\omega}\,  g(\omega+\omega_{\bm{k}'})  g(\omega+\omega_{\bm{k}})\Bigr)\, 
 a^*_{\Theta}(\bm{k}) a_{\Theta}(\bm{k}' )  
\\& \nonumber \quad \pm \Bigl(e^{-\frac{\theta}{4}\Vert k'+k \Vert^2}\,\int_{0}^{\infty} \dd{\omega}\,  g(\omega+\omega_{\bm{k}'})\, g(\omega-\omega_{\bm{k}})  \Bigr)\, a_{\Theta}(\bm{k} ) a_{\Theta} (\bm{k}' ) \\& \quad \pm\Bigl(e^{-\frac{\theta}{4}\Vert k'+k \Vert^2}\,\int_{0}^{\infty} \dd{\omega}\, g(\omega-\omega_{\bm{k}'})  g(\omega+\omega_{\bm{k}}) \Bigr)\,a^*_{\Theta} (\bm{k} )a_{\Theta}^*(\bm{k}'  )  
  \Bigr).\nonumber
\end{align}
\end{lemma}
\begin{proof}
See Appendix \ref{proflemm1}.
\end{proof}$\,$\newline Next, we consider an operator of the form
  \begin{align*}
    \int_{0}^{\infty}\dd{\omega} \,    S_{\theta}\left(    X^{\pm*}_{\omega}  X^{\pm}_{\omega} 
     \right) .
\end{align*}
In what follows we prove that this  operator is positive. 
\begin{lemma}\label{lemm2a}
    The operator  $S_{\theta}$ acting on a positive operator $X^{*}X$ is a positive operator, i.e.\
    \begin{align*}
        S_{\theta}\left(    X^{ *} X      \right) \geq 0.
    \end{align*}
\end{lemma}
\begin{proof}
For the proof we write the explicit form of the Waldmann operator acting on the positive operator $X^{ *} X$ 
\begin{align*}
      S_{\theta}\left(    X^{ *} X      \right)  &=  \pi^{-2} \int  \dd[4]{u}\,e^{-\Vert u\Vert^2}\alpha_{\sqrt{\theta}u}( X^{ *} X   )\, \,  \\&=  \pi^{-2} \int   \dd[4]{u}\,e^{-\Vert u\Vert^2} U(\sqrt{\theta}u)( X^{ *} X   )U(-\sqrt{\theta}u)\, \,    \\&=  \pi^{-2} \int   \dd[4]{u}\,e^{-\Vert u\Vert^2} U(\sqrt{\theta}u)\, X^{ *}\,U(-\sqrt{\theta}u)\,U(\sqrt{\theta}u) X   \,U(-\sqrt{\theta}u)\, \,  \\&=  \pi^{-2} \int   \dd[4]{u}\,e^{-\Vert u\Vert^2} \alpha_{\sqrt{\theta}u}(X^{ *})\,\alpha_{\sqrt{\theta}u} (X)    \, \,  \\&=  \pi^{-2} \int   \dd[4]{u}\,e^{-\Vert u\Vert^2} (\alpha_{\sqrt{\theta}u}(X))^{ *}\,\alpha_{\sqrt{\theta}u} (X)    \, \, 
\end{align*}
where in the last lines we use the fact that we  are dealing with actions of $\R^4$ that can be represented by a unitary operator $U(u)U(v)=U(u+v)$ and $U(0)=1$. Using the definition $X_u:=\alpha_{\sqrt{\theta}u}(X)$ we write 
\begin{align*}
   \langle\Psi,  S_{\theta}\left(    X^{ *} X    \right)  \Psi\rangle &= \pi^{-2} \int   \dd[4]{u}\,e^{-\Vert u\Vert^2}  \langle\Psi,  (\alpha_{\sqrt{\theta}u}(X))^{ *}\,\alpha_{\sqrt{\theta}u} (X)    \Psi\rangle  
   \\&=  \pi^{-2} \int \dd[4]{u}\,e^{-\Vert u\Vert^2}  \langle\Psi, X_u^{ *}\,X_u \Psi\rangle \, \geq0.
\end{align*}
The resulting integral is positive because both the Gaussian function and the expectation value of $X_u^{ *}\,X_u$ are positive over the entire real line. Since the product of two positive functions is positive, and the integral of a positive function over a domain with positive measure is positive, it follows that the integral is positive. In this context see as well \cite[Theorem 3.3 (1)]{waldmannpos}.
\end{proof}
\begin{lemma}\label{lemm2}
    The explicit result of the Waldmann operator applied to the product of operators $X^{\pm*}_{\omega}  X^{\pm}_{\omega}$  is
\begin{align*}
     &   \int_0^{\infty} \,\dd{\omega}\,S_{\theta}((X^{\pm}_\omega)^*X^{\pm}_\omega)
         =   W^{2,\pm}_{\Theta}+
\int_0^\infty \dd{\omega} \int  {\dd[3]{\bm{k}}}\,2 \omega_{\bm{k}}\, p(\bm{k})^2
g(\omega + \omega_{\bm{k}}) ^2  , \,  
\end{align*}
    where the operator $W^{2,\pm}_{\Theta}$ is explicitly given by 
\begin{align}\label{eqspm2}
     W^{2,\pm}_{\Theta}&:= \iint \dd[3]{\bm{k}}\dd[3]{\bm{k}'} \,  p(\bm{k})p(\bm{k}')\\
	& \nonumber
\quad \Bigl(\Bigl(e^{-\frac{\theta}{4}\Vert k'-k \Vert^2}\, \int_{0}^{\infty} \dd{\omega}\, g(\omega+\omega_{\bm{k}}) g(\omega+\omega_{\bm{k}'})\Bigr)\, e^{ik\Theta  k'}
  a^*_{\Theta} (\bm{k}' )a_{\Theta} (\bm{k}  )  \\&\nonumber
  \quad +\Bigl( e^{-\frac{\theta}{4}\Vert k'-k \Vert^2}\int_{0}^{\infty} \dd{\omega}\,  g(\omega-\omega_{\bm{k}})  g(\omega-\omega_{\bm{k}'})\Bigr)\, 
 a^*_{\Theta}(\bm{k}) a_{\Theta}(\bm{k}' )  
\\& \nonumber \quad \pm \Bigl(e^{-\frac{\theta}{4}\Vert k'+k \Vert^2}\,\int_{0}^{\infty} \dd{\omega}\,  g(\omega+\omega_{\bm{k}})\, g(\omega-\omega_{\bm{k}'})  \Bigr)\, a_{\Theta}(\bm{k} ) a_{\Theta} (\bm{k}' ) \\& \quad \pm\Bigl(e^{-\frac{\theta}{4}\Vert k'+k \Vert^2}\,\int_{0}^{\infty} \dd{\omega}\, g(\omega-\omega_{\bm{k}})  g(\omega+\omega_{\bm{k}'}) \Bigr)\, a^*_{\Theta} (\bm{k} )a_{\Theta}^*(\bm{k}'  )  
  \Bigr).\nonumber
\end{align}
\end{lemma}

\begin{proof}
    See Appendix \ref{proflemm2}.
\end{proof}$\,$\newline
Finally, combining the two positive operators constructed in the previous lemmas through a linear combination, we obtain an operator that matches the smeared, normal-ordered energy density defined by:
\begin{align}\label{eqspmW}\boxed{
  W_{\Theta}^{\pm}=\frac{1}{2} \int_{0}^{\infty}\dd{\omega} \,    S_{\theta}\left(    X^{\pm*}_{\omega} \times_{\theta}X^{\pm}_{\omega} 
     \right) +   \frac{1}{2} \int_{0}^{\infty}\dd{\omega} \,    S_{\theta}\left(    X^{\pm*}_{\omega}  X^{\pm}_{\omega} 
     \right).   }
\end{align}When the deformation parameter 
$\theta$ is set to zero and the second term is symmetrized in the momentum variables $\bm{k}$ and $\bm{k}'$ the operator reduces to the form studied in \cite{Few98}. 
\begin{proposition} \label{lemm3}
  The operator $W_{\Theta}^{\pm}$ given as a  linear combination in Equation \eqref{eqspmW} is explicitly given by 
  \begin{align*} W_{\Theta}^{\pm}&=\frac{1}{2}\iint \dd[3]{\bm{k}}'\dd[3]{\bm{k}} \,  p(\bm{k})p(\bm{k}') \Bigl(
F_{\theta}(\bm{k},\bm{k}')\, \Bigl( e^{ik\Theta  k'}
  a^*_{\Theta} (\bm{k}' )a_{\Theta} (\bm{k}  ) +
 a^*_{\Theta}(\bm{k}) \,a_{\Theta}(\bm{k}' )  \Bigr) 
  \\& \quad \pm  {G}_{\theta}(\bm{k},\bm{k}')\, \Bigl(a_{\Theta}(\bm{k} ) a_{\Theta} (\bm{k}' )
    +a^*_{\Theta} (\bm{k} )a_{\Theta}^*(\bm{k}'  )\Bigr) \Bigr) +
\int_0^\infty \dd{\omega} \int  {\dd[3]{\bm{k}}} \,2\omega_{\bm{k}}\,  p(\bm{k})^2
g(\omega + \omega_{\bm{k}}) ^2  ,\, 
\end{align*}
with functions $F$ and $G$ that read
\begin{align*}
    F_{\theta}(\bm{k},\bm{k}')&= 
e^{-\frac{\theta}{4}\Vert k'-k \Vert^2}\,  \int_{0}^{\infty} \dd{\omega}\, \Bigl(g(\omega-\omega_{\bm{k}'})\, g(\omega-\omega_{\bm{k}})+g(\omega+\omega_{\bm{k}'})\, g(\omega+\omega_{\bm{k}})\Bigr), \\
 {G}_{\theta}(\bm{k},\bm{k}')&= e^{-\frac{\theta}{4}\Vert k'+k \Vert^2}\,\int_{0}^{\infty} \dd{\omega}\,  \Bigl(g(\omega+\omega_{\bm{k}'})\, g(\omega-\omega_{\bm{k}}) + g(\omega-\omega_{\bm{k}'})\, g(\omega+\omega_{\bm{k}})\Bigr).
\end{align*}

\end{proposition}
\begin{proof}
    See Appendix \ref{proflemm3}. 
\end{proof}$\,$\newline
Analogously to \cite{Few98} we use    the fact that the function $g$ is even, to obtain 
\begin{align}\nonumber
    F_{\theta}(\bm{k},\bm{k}')&=  
e^{-\frac{\theta}{4}\Vert k'-k \Vert^2}\,  \int_{-\infty}^{\infty} \dd{\omega}\, g(\omega-\omega_{\bm{k}'}) g(\omega-\omega_{\bm{k}})  \\&= e^{-\frac{\theta}{4}\Vert k'-k \Vert^2}\,(g* g)(\omega_{\bm{k}}-\omega_{\bm{k}'}) \nonumber \\ 
&= \widehat{f_{\theta}}(\omega_{\bm{k}}-\omega_{\bm{k}'}, \bm{k}-\bm{k}'),  \label{eqftheta} 
\end{align}
where the convolution of two functions \( f \) and \( g \) is defined as:
\[
(f * g)(t) = \int_{-\infty}^{\infty}\dd{\tau} f(\tau)\, g(t-\tau).
\]
Equivalently for the function ${G}_{\theta}$ we have
\begin{align}
 {G}_{\theta}(\bm{k},\bm{k}')&= e^{-\frac{\theta}{4}\Vert k'+k \Vert^2}\,(g* g)(\omega_{\bm{k}}+\omega_{\bm{k}'}) \nonumber\\ 
 &= \widehat{f_{\theta}}(\omega_{\bm{k}}+\omega_{\bm{k}'}, \bm{k}'+\bm{k}) , \label{eqgtheta} 
 \end{align} 
 where in the last equations we identified the functions $F_{\theta}$ and ${G}_{\theta}$ with  the Fourier transformed  smearing function $\widehat{f_{\theta}}$ that appears in Equation \eqref{eq:rdemt}.

 Since we constructed the operator $W^{\pm}_{\Theta}$ as formally positive, we have the state independent inequality
\begin{align}\boxed{
   \frac{1}{2}  \left\langle W^{1,\pm}_{\Theta}+W^{2,\pm}_{\Theta} \right\rangle_{\Psi} \geq - 
\int_0^\infty \dd{\omega} \int  {\dd[3]{\bm{k}}} \, 2 \omega_{\bm{k}}\,p(\bm{k})^2
g(\omega + \omega_{\bm{k}}) ^2\,  .}\label{eqboundW}
\end{align} 

 \section{The Non-commutative QEI}
The  non-trivial commutation relation between the time and the space coordinates, prevent us physically to consider  the energy density at position $\bm{x}=\bm{0}$, and smear with a function  in time only. However, to keep it close to the original framework, we smear with a test function that behaves like the one used in \cite{Few98}, when $\Theta$ goes to zero.  
Next, we compare the operator $W^{\pm}_{\Theta}$ given in   Equation \eqref{eqspmW}  to the terms in the normal-ordered smeared energy density given explicitly in Equation \eqref{eq:rdemt}. Without loss of generality, we split the test function $\hat{f}_{\theta}$ (see Equation \eqref{eqffh}) as follows  
\begin{align*}
 \widehat{f_{\theta} }(\omega_{\bm{k}}, \bm{k})= \widehat{f_{\theta} }(\omega_{\bm{k}}) \widehat{h_{\theta}}(\bm{k}),
\end{align*}
and make the following identifications using Proposition \ref{lemm3} and Equations \eqref{eqftheta} and \eqref{eqgtheta},
\begin{align}\label{eqfhat}
     \widehat{f_{\theta} }(\omega_{\bm{k}}  \pm  \omega_{\bm{k}'}) &=e^{-\frac{\theta}{4} (\omega_{\bm{k}'} \pm \omega_{\bm{k}})^2}\, \left(g* g\right)( \omega_{\bm{k}} \pm \omega_{\bm{k}'} )\\
      \widehat{h_{\theta}} ( {\bm{k}} \pm  {\bm{k}'})&=e^{-\frac{\theta}{4}\Vert\bm{k}'\pm\bm{k }\Vert^2}.\,
\end{align}
To match the test-functions $\hat{f}_{\theta}$ and the functions appearing in the operator $W^{\pm}_{\Theta}$, we set
\begin{align}
   \widehat{f_{\theta} }(\cdot)=  e^{-\frac{\theta}{4} (\cdot)} \hat{f}(\cdot) ,
\end{align}
where $\hat{f}=g * g$ such that $\widehat{f^{1/2}}=g$. The inverse Fourier transformation of the function $\widehat{f_{\theta}}$ reads  
\begin{align}
f_{\theta}(t)&=\left(K_{\theta} * f\right)(t) 
\\&=
\int_{-\infty}^{\infty} \dd{\tau} \frac{1}{\sqrt{\theta}}e^{-\frac{(t-\tau)^2}{\theta}} f(\tau).
\end{align}
where we defined the Kernel $K_{\theta}$
\[
K_{\theta}(t) = \frac{1}{\sqrt{\theta}} \, e^{-\frac{t^2}{\theta}}.
\]
 Note that in the commutative limit \( \theta \to 0 \), the kernel \( K_{\theta}(t) \) converges to the Dirac delta distribution \( \delta(t) \), and consequently, the deformed function \( f_{\theta}(t) \) reduces to the original function \( f(t) \). \par
 Therefore the quantum inequality that we prove holds for all test-functions $f$ where $f^{1/2}$ is a test-function as well, as in the original proof in \cite{Few98}, after the convolution thereof with the Kernel $K_{\theta}$   that acts as a localized weighting function. Physically, this convolution generates a smoothed ("smeared") version of $f(t)$, with the parameter $\theta$ controlling the  scale of smoothing. Larger $\theta$ values increase the effective averaging window, suppressing high-frequency features while preserving low-frequency trends. The process is equivalent to evaluating $f(t)$ at points distributed normally around $t$ with variance $\theta/2$, effectively modeling diffusion-like spreading of information. In Fourier space, this corresponds to a low-pass filter that attenuates frequencies $k$ via the factor $e^{-\frac{\theta}{4} \Vert k\Vert^2}$. Mathematically, $f_\theta(t)$ satisfies the heat equation with initial condition $f(t)$ and diffusion time $\theta/2$, formalizing its role as a regularization operation. Beyond their physical interpretation in the context of quantum geometry, such $\theta$-dependent Gaussians frequently appear in the literature. In numerous examples involving non-commutative black holes, mass distributions are replaced by these "deformed Gaussians," which serve as natural regularization functions for point masses (see \cite{Gauss2} (and references therein), \cite{Gauss3, Gauss4, Gauss5}).  The kernel $K_{\theta}$ in space and time has also been utilized in the context of quantum field theory (QFT), where it replaces the distribution in equal-time canonical commutation relations. This approach is then applied to various problems in cosmology, see \cite[Section 5]{Gauss1}.
\par
Since the deformed, normal-ordered and smeared non-commutative energy density is a finite sum of operators $W^{1,\pm}_{\Theta}$ and $W^{2,\pm}_{\Theta}$, we apply analogously to \cite{Few98} the Bound \eqref{eqboundW} for the different cases 
\begin{align*}
    p(\bm{k})=\frac{1}{2},\qquad \qquad 
    p(\bm{k})=\frac{k_i}{ 2{ \omega_{\bm{k}}}},\qquad \qquad  p(\bm{k})=\frac{m}{2{\omega_{\bm{k}}}}, 
\end{align*}
and obtain, on summing, the commutative inequality 
\begin{align*}
   \left\langle \ord{T^\Theta_{f_\theta}} \right\rangle_\Psi &\geq - \int_{0}^{\infty}\dd{\omega} \int \dd[3]{\bm{k}} \,  \omega_{\bm{k}}\, g(\omega+\omega_{\bm{k}})^2 \\ 
   &= - \int_{0}^{\infty}\dd{\omega} \int \dd[3]{\bm{k}} \, \omega_{\bm{k}}
\widehat{f^{1/2}}(\omega+\omega_{\bm{k}})^2\\&=
- \frac{C_3}{2\pi}\int_{0}^{\infty}\dd{\omega}\int_{m}^{\infty} \dd{\omega}'\, 
\widehat{f^{1/2}}(\omega+\omega')^2\,{\omega}'^2\,\left( {\omega}'^2-m^2 \right)^{1/2},
\end{align*}
where $C_3$ is a constant (see \cite[Equation (4.5) for  $n=3$]{Few98}. This  inequality is the same inequality as in the undeformed case and it is finite in case $f^{1/2}$ is a test-function\footnote{This condition is sufficient but not necessary.}. The necessary and sufficient condition is exactly that the decay of the Fourier transform of \( f^{1/2} \) beats the polynomial weight \( \omega'^3 \). This is guaranteed by assuming that \( f \) is nonnegative and that \( f^{1/2} \) is sufficiently smooth and decays rapidly—for instance, if \( f \) belongs to the Schwartz space or, equivalently, if \( \sqrt{f} \in H^s(\mathbb{R}) \) (the Sobolev space) with \( s > 2\).

\section{Conclusion and Outlook}
In this work, we established a quantum energy inequality (QEI) for massive scalar fields in non-commutative Minkowski spacetime. The derived lower bound retains independence from the non-commutative length scale and coincides exactly with its commutative counterpart, demonstrating two critical physical implications: first, the absence of causality violations in this minimally coupled quantum field theory, despite the underlying non-commutative geometry; and second, the vanishing of quantum geometry corrections for spacetime-averaged observables, effectively recovering classical locality in the macroscopic regime. These results suggest that non-commutative modifications at small scales need not destabilize causal structure or perturb large-scale physics, reinforcing theoretical confidence in such models as consistent bridges between quantum geometry and semiclassical phenomenology.\par
Furthermore, a quantum weak energy inequality is an indication of the mesoscopic stability of the quantum system, and is linked to microscopic and macroscopic stability \cite{Fewster_2003b, Fewster_2007a}. In particular, microscopic stability is given by the microlocal spectrum condition, which is a characterization of Hadamard states using the wavefront sets of $n$-point distributions \cite{RAD2, SV01, IN3, SVW, Sanders}. Fewster pioneered the application of microlocal analysis to quantum energy inequalities \cite{Fewster_2000, Fewster_2002, Fewster_2003a, Fewster_2006, Fewster_2008}, \textit{i.e.}\ the link from microscopic to mesoscopic stability. This entailed the use of the microlocal spectrum condition in restricting the energy density to the worldline of an observer, and in ensuring the convergence of the integrals in the energy inequalities. In light of the recent proof of the microlocal spectrum condition for states on an algebra of field operators deformed by warped convolutions \cite{BMV}, we plan a similar application of microlocal analysis in deriving another quantum weak energy inequality (which might very well turn out to be equivalent). Since non-commutative Minkowski spacetime does not admit the concept of a point, and thus of a worldline, it only makes sense to consider the existence of a quantum energy inequality over a worldvolume. However, the deformed energy density can be restricted to the worldline of an observer in commutative Minkowski spacetime, for which a quantum energy inequality might still exist; we anticipate that the formalism of wavefront sets will provide mathematical evidence against the existence of a quantum energy inequality along a worldline for a deformed quantum field theory.

\section*{Acknowledgments}
The authors would like to thank Rainer Verch, Markus Fröb and Philipp Dorau for various discussions. Furthermore, we extend our deepest appreciation to Rishabh Ballal for his substantial input and thoughtful engagement throughout the development and during the last stages of this work.
\appendix
\section{Normal Ordering}\label{appnorm}

The exponential factor including the momentum operator which induces the deformation can be written as a normal ordered term, using results from \cite{G}, namely
\begin{align*}
e^{ik\Theta P}=    \ord{e^{\int \frac{\dd[3]\bm{k}}{2\omega_{\bm{k}}}\,(e^{ip\Theta k}-1)\,a^*  (\bm{k} )a (\bm{k} )}}.
\end{align*}
Next, we turn to the question of the normal ordering of the deformed creation and annihilation operators. First, we consider 
\begin{align*}
 a_{\Theta } (\bm{k} )a^*_{\Theta }  (\bm{k}') &=
 e^{\frac{i}{2}k\Theta P}a (\bm{k} )e^{-\frac{i}{2}k'\Theta P}a^*(\bm{k}')  \\&
 = e^{-\frac{i}{2}k'\Theta k}e^{\frac{i}{2}(k-k')\Theta P}a^*(\bm{k}') a (\bm{k} )+ 2{ \omega_{\bm{k}}}\delta^3(\bm{k}-\bm{k}')\\&
 = e^{-\frac{i}{2}k'\Theta k}e^{\frac{i}{2}(k-k')\Theta k'} a^*(\bm{k}') e^{\frac{i}{2}(k-k')\Theta P} a (\bm{k} )+ 2{ \omega_{\bm{k}}}\delta^3(\bm{k}-\bm{k}')\\&
 = e^{ i k\Theta k'} a^*(\bm{k}') \ord{e^{\int \frac{\dd[3]\bm{p}}{2\omega_{\bm{p}}}\,(e^{i(k-k')\Theta p}-1)\,a^*  (\bm{p} )a(\bm{p})}}  a (\bm{k} )+ 2{ \omega_{\bm{k}}}\delta^3(\bm{k}-\bm{k}').
\end{align*}
The normal ordering of the former term is 
\begin{align*}
 \ord{a_{\Theta } (\bm{k} )a^*_{\Theta }  (\bm{k}')} &= a_{\Theta } (\bm{k} )a^*_{\Theta }  (\bm{k}')-\langle \Omega\vert  a_{\Theta } (\bm{k} )a^*_{\Theta }  (\bm{k}')\Omega\rangle\\& =e^{ i k\Theta k'} a^*(\bm{k}') \ord{e^{\int \frac{\dd[3]\bm{p}}{2\omega_{\bm{p}}}\,(e^{i(k-k')\Theta p}-1)\,a^*  (\bm{p} )a(\bm{p})}}  a (\bm{k} ),
\end{align*}
and matches with the naive normal ordering rule
\begin{align*}
    \ord{a_{\Theta } (\bm{k} )a^*_{\Theta }  (\bm{k}')} =e^{ik\Theta k'}a^*_{\Theta }  (\bm{k}')a_{\Theta } (\bm{k} ).
\end{align*}
This becomes obvious when writing out the right hand side 
\begin{align*}
    e^{ik\Theta k'}a^*_{\Theta }  (\bm{k}')a_{\Theta } (\bm{k} )&=  e^{ik\Theta k'}  e^{-\frac{i}{2}k'\Theta P} a^*   (\bm{k}')e^{ \frac{i}{2}k \Theta P} a (\bm{k} )\\
    &=  e^{ik\Theta k'}  a^*   (\bm{k}') e^{-\frac{i}{2}k'\Theta P} e^{ \frac{i}{2}k \Theta P} a (\bm{k} )\\
    &=  e^{ik\Theta k'}  a^*   (\bm{k}') e^{-\frac{i}{2}(k-k')\Theta P}   a (\bm{k} )\\&=e^{ i k\Theta k'} a^*(\bm{k}') \ord{e^{\int \frac{\dd[3]\bm{p}}{2\omega_{\bm{p}}}\,(e^{i(k-k')\Theta p}-1)\,a^*  (\bm{p} )a(\bm{p})}}  a (\bm{k} ),
\end{align*}
where in the last lines we used the commutativity $a^*(\bm{k}') e^{-\frac{i}{2}k'\Theta P}= e^{-\frac{i}{2}k'\Theta P}a^*(\bm{k}')$ due to the skew-symmetry of $\Theta$.

\section{Concrete Calculations}

\subsection{Deformed Products and the Waldmann operator}\label{app1}
In this section we supply the concrete derivations of the deformed product and application of the operator $S_{\theta}$ to the deformed annihilation and creation  operators. The deformed product of the deformed annihilation and creation operators is given by
\begin{align*}
    a_{\Theta}^*(\bm{k}')\times_{\theta}a_{\Theta} (\bm{k})&= (2\pi)^{-4} \iint \dd[4]{x}\,\dd[4]{y}\, e^{-ixy}
    \alpha_{\Theta x}(a^*_{\Theta}(\bm{k}'))  \alpha_{y}(a_{\Theta}(\bm{k} ))
\\&= (2\pi)^{-4}
\iint \dd[4]{x}\,\dd[4]{y}\, e^{-ixy}\, 
e^{ik'\Theta x} e^{-ik y} \,a^*_{\Theta}(\bm{k}')\,a_{\Theta}(\bm{k} ) \\&= 
\int \dd[4]{x}\,  
e^{ik'\Theta x} \delta(x+k )  \,a^*_{\Theta}(\bm{k}')\,a_{\Theta}(\bm{k} )\\&= 
e^{ik\Theta k'}  \,a^*_{\Theta}(\bm{k}')\,a_{\Theta}(\bm{k} ).
\end{align*}
Next, we calculate the action of the operator $S_{ {\theta} }$, i.e.\ 
\begin{align}\nonumber
   S_{\theta}\left(a_{\Theta}^*(\bm{k}') a_{\Theta} (\bm{k} )\right)&=  \pi^{-2} \int \dd[4]{u} \,e^{-\Vert u\Vert^2}\alpha_{\sqrt{\theta}u}(a_{\Theta}^*(\bm{k}') a_{\Theta} (\bm{k} ))\, 
\\& \nonumber= \pi^{-2} \int \dd[4]{u} \,e^{-\Vert u\Vert^2}\alpha_{\sqrt{\theta}u}(a_{\Theta}^*(\bm{k}'))\alpha_{\sqrt{\theta}u}( a_{\Theta} (\bm{k} ))\,  \\& \nonumber= \pi^{-2} \int  \dd[4]{u}\,e^{-\Vert u\Vert^2}
e^{ik' \sqrt{\theta}u} e^{-ik \sqrt{\theta}u} 
  \, a_{\Theta}^*(\bm{k}')   a_{\Theta} (\bm{k} ) \\& \nonumber= \pi^{-2} \int  \dd[4]{u}\,e^{-\Vert u\Vert^2}
e^{i u \sqrt{\theta}(k'-k)} 
\,  a_{\Theta}^*(\bm{k}')   a_{\Theta} (\bm{k} ) \\& =  e^{-\frac{\theta}{4}\Vert k'-k \Vert^2} \,  a_{\Theta}^*(\bm{k}')   a_{\Theta} (\bm{k} ). \label{eqsaastar}
\end{align}
The deformed product of the deformed operators is given by
\begin{align*}
   a_{\Theta} (\bm{k}' )\times_{\theta}a_{\Theta}^*(\bm{k})&= (2\pi)^{-4} \iint \dd[4]{x}\,\dd[4]{y}\, e^{-ixy}
    \alpha_{\Theta x}( a_{\Theta} (\bm{k}' ))  \alpha_{y}(a_{\Theta}^*(\bm{k}))
\\&= (2\pi)^{-4}
\iint \dd[4]{x}\,\dd[4]{y}\, e^{-ixy}\, 
e^{-ik'\Theta x} e^{iky} \,a_{\Theta}(\bm{k}' ) \, a^*_{\Theta}(\bm{k}) \\&=
\int \dd[4]{x}\,  
e^{-ik' \Theta x} \delta(x-k )  \,a_{\Theta}(\bm{k}' ) \, a^*_{\Theta}(\bm{k})\\&= 
e^{-ik'\Theta k}  \,a_{\Theta}(\bm{k}' ) \, a^*_{\Theta}(\bm{k})\\&= 
e^{-ik'\Theta k} e^{ik'\Theta k} \,a^*_{\Theta}(\bm{k}) \,a_{\Theta}(\bm{k}' ) +  2\omega_{\bm{k}} \,e^{-ik'\Theta k}\delta^3(\bm{k}'-\bm{k})  \\&= 
 a^*_{\Theta}(\bm{k}) \,a_{\Theta}(\bm{k}') +2 \omega_{\bm{k}}  \delta^3(\bm{k}'-\bm{k}).
\end{align*}
The action of the operator $S_{ {\theta} }$ on this deformed product is, 
\begin{align*}
   S_{\theta}\left( a_{\Theta} (\bm{k}')a_{\Theta}^*(\bm{k})  \right)&=     e^{-\frac{\theta}{4}\Vert k'-k \Vert^2} \,  a_{\Theta} (\bm{k}') a_{\Theta}^*(\bm{k}) ,
\end{align*}
where this follows directly from taking the adjoint of Equation \eqref{eqsaastar}. The deformed product of the deformed annihilation operators is given by
\begin{align*}
   a_{\Theta} (\bm{k}' )\times_{\theta}a_{\Theta} (\bm{k})&= (2\pi)^{-4} \iint \dd[4]{x}\,\dd[4]{y}\, e^{-ixy}
    \alpha_{\Theta x}( a_{\Theta} (\bm{k}' ))  \alpha_{y}(a_{\Theta} (\bm{k}))
\\&= (2\pi)^{-4} \iint \dd[4]{x}\,\dd[4]{y}\, e^{-ixy}
e^{-ik'\Theta x}e^{-iky} \, a_{\Theta} (\bm{k}' )   a_{\Theta} (\bm{k}) \\&= \int \dd[4]{x}\, e^{-ik'\Theta x} \delta(x+k)\,
 a_{\Theta} (\bm{k}' )   a_{\Theta} (\bm{k}) \\&=  
e^{ik'\Theta k } \, a_{\Theta} (\bm{k}' )   a_{\Theta} (\bm{k}) \\&=  
e^{ ik'\Theta k }e^{ -ik'\Theta k } \, a_{\Theta} (\bm{k})   a_{\Theta} (\bm{k}') \\&=  
  a_{\Theta} (\bm{k})   a_{\Theta} (\bm{k}') .
\end{align*}
Next, we calculate the action of the operator $S_{ {\theta} }$, i.e.\ 
\begin{align}\nonumber
   S_{\theta}\left( a_{\Theta} (\bm{k}' )a_{\Theta} (\bm{k})  \right)&= \pi^{-2} \int   \dd[4]{u}\,e^{-\Vert u\Vert^2}\alpha_{\sqrt{\theta}u}(a_{\Theta} (\bm{k}' )a_{\Theta} (\bm{k}) )\,
\\&\nonumber = \pi^{-2} \int  \dd[4]{u} \,e^{-\Vert u\Vert^2}\alpha_{\sqrt{\theta}u}(a_{\Theta} (\bm{k}'))\alpha_{\sqrt{\theta}u}( a_{\Theta} (\bm{k}))\,  \\&\nonumber = \pi^{-2} \int  \dd[4]{u} \,e^{-\Vert u\Vert^2}
e^{-ik' \sqrt{\theta}u} e^{-ik \sqrt{\theta}u} \,   (a_{\Theta} (\bm{k}')  a_{\Theta} (\bm{k}))\\&\nonumber = \pi^{-2} \int   \dd[4]{u}\,e^{-\Vert u\Vert^2}
e^{i u \sqrt{\theta}(k'+k)}  
\,    a_{\Theta} (\bm{k}' ) a_{\Theta}(\bm{k})  \\&   =  e^{-\frac{\theta}{4}\Vert k'+k \Vert^2} \,  a_{\Theta} (\bm{k}') a_{\Theta} (\bm{k}). \label{eqsaastar1}
\end{align}
The deformed product of the deformed creation operators is given by
\begin{align*}
   a^{*}_{\Theta} (\bm{k}')\times_{\theta}a^{*}_{\Theta} (\bm{k})&= (2\pi)^{-4} \iint \dd[4]{x}\,\dd[4]{y}\, e^{-ixy}
    \alpha_{\Theta x}( a^{*}_{\Theta} (\bm{k}'))  \alpha_{y}(a^{*}_{\Theta} (\bm{k}))
\\&= (2\pi)^{-4} \iint \dd[4]{x}\,\dd[4]{y}\, e^{-ixy}
e^{ ik'\Theta x}e^{ iky} \, a^{*}_{\Theta} (\bm{k}')   a^{*}_{\Theta} (\bm{k}) \\&= \int \dd[4]{x}\, \delta(x-k)
e^{ ik'\Theta x} \,  a^{*}_{\Theta} (\bm{k}')   a^{*}_{\Theta} (\bm{k}) \\&=   e^{ ik'\Theta k} \,  a^{*}_{\Theta} (\bm{k}')   a^{*}_{\Theta} (\bm{k}) \\&=   e^{ ik'\Theta k} e^{ -ik'\Theta k} \,  a^{*}_{\Theta} (\bm{k})   a^{*}_{\Theta} (\bm{k}') \\&=     a^{*}_{\Theta} (\bm{k})   a^{*}_{\Theta} (\bm{k}') .\end{align*}
The action of the operator $S_{ {\theta} }$ follows directly by taking the adjoint of Equation \eqref{eqsaastar1}:
\begin{align*}
   S_{\theta}\left( a^{*}_{\Theta} (\bm{k}')a^{*}_{\Theta} (\bm{k})  \right)&=        e^{-\frac{\theta}{4}\Vert k'+k \Vert^2} \,  a^{*}_{\Theta} (\bm{k}') a^{*}_{\Theta} (\bm{k}) .
\end{align*} 

\subsection{Proof of Lemma \ref{lemm1} }\label{proflemm1}
\begin{proof} 
 Let us write the product $X^{\pm*}_{\omega} \times_{\theta}X^{\pm}_{\omega}$ in a more explicit form, 
    \begin{align*}& 
  X^{\pm*}_{\omega} \times_{\theta}X^{\pm}_{\omega}=    
     \iint \dd[3]{\bm{k}'} \,p(\bm{k}')
     \dd[3]{\bm{k}}  \,p(\bm{k})\\&
     \quad \Bigl(\left(g(\omega-\omega_{\bm{k}'})a^*_{\Theta} (\bm{k}')\pm
  g(\omega+\omega_{\bm{k}'})a_{\Theta}(\bm{k}')\right) \times_{\theta}  \left(  g(\omega-\omega_{\bm{k}})a_{\Theta} (\bm{k} )\pm
    g(\omega+\omega_{\bm{k}})a_{\Theta}^*(\bm{k})\right) \Bigr)
\\&=   \iint \dd[3]{\bm{k}'}\dd[3]{\bm{k}} \,  p(\bm{k})p(\bm{k}')\\&
\quad \Bigl(  g(\omega-\omega_{\bm{k}'}) g(\omega-\omega_{\bm{k}})a^*_{\Theta} (\bm{k}')\times_{\theta} a_{\Theta} (\bm{k})\pm g(\omega-\omega_{\bm{k}'})  g(\omega+\omega_{\bm{k}})a^*_{\Theta} (\bm{k}')\times_{\theta}a_{\Theta}^*(\bm{k})
\\&\quad \pm g(\omega+\omega_{\bm{k}'})\, g(\omega-\omega_{\bm{k}})a_{\Theta}(\bm{k}')\times_{\theta}a_{\Theta} (\bm{k})
+ g(\omega+\omega_{\bm{k}'}) g(\omega+\omega_{\bm{k}}) a_{\Theta}(\bm{k}')  \times_{\theta}a_{\Theta}^*(\bm{k})\Bigr) \\&=   \iint \dd[3]{\bm{k}'}\dd[3]{\bm{k}} \,  p(\bm{k})p(\bm{k}')
\Bigl(  g(\omega-\omega_{\bm{k}'}) g(\omega-\omega_{\bm{k}})e^{ik\Theta k'}
a^*_{\Theta} (\bm{k}')a_{\Theta} (\bm{k})
\\&\quad \pm g(\omega-\omega_{\bm{k}'})  g(\omega+\omega_{\bm{k}})a^*_{\Theta} (\bm{k} )a_{\Theta}^*(\bm{k}' )
 \pm g(\omega+\omega_{\bm{k}'})\, g(\omega-\omega_{\bm{k}})a_{\Theta}(\bm{k} ) a_{\Theta} (\bm{k}')
\\& \quad + g(\omega+\omega_{\bm{k}'}) g(\omega+\omega_{\bm{k}})\left(
 a^*_{\Theta}(\bm{k}) \,a_{\Theta}(\bm{k}') +2\omega_{k}  \delta^3(\bm{k}-\bm{k}')\right)\Bigr).
\end{align*}
Next, we apply the operator $S_{\theta}$ rendering 
\begin{align*}
   & S_{\theta}\left(   X^{\pm*}_{\omega} \times_{\theta}X^{\pm}_{\omega} 
     \right) =\iint \dd[3]{\bm{k}'}\dd[3]{\bm{k}} \,  p(\bm{k})p(\bm{k}')\\&
\Bigl(  g(\omega-\omega_{\bm{k}'}) g(\omega-\omega_{\bm{k}})e^{ik\Theta k'}
  S_{\theta}\left( a^*_{\Theta} (\bm{k}')a_{\Theta} (\bm{k}) \right)
  \pm g(\omega-\omega_{\bm{k}'})  g(\omega+\omega_{\bm{k}})  S_{\theta}\left( a^*_{\Theta} (\bm{k} )a_{\Theta}^*(\bm{k}' ) \right)\\&
 \pm  g(\omega+\omega_{\bm{k}'})\, g(\omega-\omega_{\bm{k}}) S_{\theta}\left(  a_{\Theta}(\bm{k} ) a_{\Theta} (\bm{k}') \right)
  + g(\omega+\omega_{\bm{k}'}) g(\omega+\omega_{\bm{k}})  S_{\theta}\left( 
 a^*_{\Theta}(\bm{k}) \,a_{\Theta}(\bm{k}')  \right)\Bigr)
\\&+ \int \dd[3]{\bm{k}} \, 2\omega_{k}\, p(\bm{k})^2
g(\omega + \omega_{\bm{k}}) ^2 \, \\&  
 =\iint \dd[3]{\bm{k}'} \dd[3]{\bm{k}} \,  p(\bm{k})p(\bm{k}') 
\Bigl(  g(\omega-\omega_{\bm{k}'}) g(\omega-\omega_{\bm{k}})e^{ik\Theta k'}
 e^{-\frac{\theta}{4}\Vert k'-k \Vert^2}\, a^*_{\Theta} (\bm{k}')a_{\Theta} (\bm{k})  
\\&\quad \pm g(\omega-\omega_{\bm{k}'})  g(\omega+\omega_{\bm{k}}) e^{-\frac{\theta}{4}\Vert k'+k \Vert^2}\,a^*_{\Theta} (\bm{k} )a_{\Theta}^*(\bm{k}' )  
 \\&\quad \pm  g(\omega+\omega_{\bm{k}'})\, g(\omega-\omega_{\bm{k}}) e^{-\frac{\theta}{4}\Vert k'+k \Vert^2}\, a_{\Theta}(\bm{k} ) a_{\Theta} (\bm{k}')  
\\&\quad + g(\omega+\omega_{\bm{k}'}) g(\omega+\omega_{\bm{k}}) e^{-\frac{\theta}{4}\Vert k'-k \Vert^2}\, 
 a^*_{\Theta}(\bm{k}) \,a_{\Theta}(\bm{k}')  \Bigr)
+ \int \dd[3]{\bm{k}} \,  2\omega_{k}\,  p(\bm{k})^2
g(\omega + \omega_{\bm{k}}) ^2,  \, 
\end{align*} 
where the former application of the deformed product and the operator $S_{\theta}$ is calculated term for term in Appendix \ref{app1}.  An integral over $\omega$  renders the result.\end{proof}

\subsection{Proof of Lemma \ref{lemm2}}\label{proflemm2}

\begin{proof}
Using the commutation relations we have
\begin{align*}
    (X^{\pm}_\omega)^*X^{\pm}_\omega
        &= \iint  {\dd[3]{\bm{k}}\, p(\bm{k})\, \dd[3]{\bm{k}'}}\, p(\bm{k}') \\&\quad \Bigl(\e^{i k\Theta k'} {g(\omega + \omega_{\bm{k}})} g(\omega + \omega_{\bm{k}'}) a^*_\Theta (\bm{k}') a_\Theta (\bm{k}) + {g(\omega - \omega_{\bm{k}})}  {g(\omega - \omega_{\bm{k}'})} a_\Theta^* (\bm{k}) a_\Theta(\bm{k}') 
       \\&\quad \pm  {g(\omega + \omega_{\bm{k}})}       g(\omega - \omega_{\bm{k}'}) a_\Theta(\bm{k})  a_\Theta (\bm{k}') 
       \pm {g(\omega - \omega_{\bm{k}})} \, {g(\omega + \omega_{\bm{k}'})} a_\Theta^* (\bm{k})a_\Theta^* (\bm{k}') \Bigr)
        \\&\quad + \int_0^\infty \dd{\omega} \int \dd[3]{\bm{k}}  \,  2\omega_{k}\, p(\bm{k})^2
        g(\omega + \omega_{\bm{k}}) ^2  \,,
\end{align*}
and, on applying the Waldmann operator $S_{\theta}$, we obtain 
\begin{align*}
     &   \int_0^{\infty} \,\dd{\omega}\,S_{\theta}((X^{\pm}_\omega)^*X^{\pm}_\omega)
        = \iint  {\dd[3]{\bm{k}} \dd[3]{\bm{k}'}}\, p(\bm{k})
       p(\bm{k}')    \\&  \quad \int_0^{\infty} \,\dd{\omega}\,
       \Bigl(\e^{\i k \Theta k'}{g(\omega + \omega_{\bm{k}})}  {g(\omega + \omega_{\bm{k}'})} S_{\theta}(a_\Theta^* (\bm{k}') a_\Theta(\bm{k})) +  {g(\omega - \omega_{\bm{k}})} g(\omega - \omega_{\bm{k}'}) S_{\theta}(a^*_\Theta (\bm{k}) a_\Theta (\bm{k}') ) \\&
        \quad \pm  {g(\omega + \omega_{\bm{k}})}       g(\omega - \omega_{\bm{k}'})S_{\theta}(a_\Theta(\bm{k})  a_\Theta (\bm{k}') )
       \pm {g(\omega - \omega_{\bm{k}})} \, {g(\omega + \omega_{\bm{k}'})} S_{\theta}(a_\Theta (\bm{k})a_\Theta (\bm{k}') ) \Bigr)
       \\& \quad +
\int_0^\infty \dd{\omega} \int  \dd[3]{\bm{k}} \,   2\omega_{k}\,p(\bm{k})^2
g(\omega + \omega_{\bm{k}}) ^2  \, 
       \\&=    \iint  {\dd[3]{\bm{k}} \dd[3]{\bm{k}'}} p(\bm{k})
       p(\bm{k}')   \int_0^{\infty} \,\dd{\omega}\,
       \Bigl(\e^{\i k \Theta k'}{g(\omega + \omega_{\bm{k}})}  {g(\omega + \omega_{\bm{k}'})} e^{-\frac{\theta}{4}\Vert k'-k \Vert^2}(a_\Theta^* (\bm{k}')a_\Theta(\bm{k}))  \\& \quad +  {g(\omega - \omega_{\bm{k}})} g(\omega - \omega_{\bm{k}'}) e^{-\frac{\theta}{4}\Vert k'-k \Vert^2}(a^*_\Theta (\bm{k}) a_\Theta (\bm{k}') ) \\&
        \quad \pm  {g(\omega + \omega_{\bm{k}})}       g(\omega - \omega_{\bm{k}'})e^{-\frac{\theta}{4}\Vert k'+k \Vert^2}(a_\Theta(\bm{k})  a_\Theta (\bm{k}') )
   \\& \quad \pm {g(\omega - \omega_{\bm{k}})} \, {g(\omega + \omega_{\bm{k}'})} e^{-\frac{\theta}{4}\Vert k'+k \Vert^2}(a^* _\Theta (\bm{k})a_\Theta^* (\bm{k}') ) \Bigr)
     +
\int_0^\infty \dd{\omega} \int  {\dd[3]{\bm{k}}}  \, 2\omega_{k}\,  p(\bm{k})^2
g(\omega + \omega_{\bm{k}}) ^2  \, 
       \\&=   W^{2,\pm}_{\Theta}+
\int_0^\infty \dd{\omega} \int  {\dd[3]{\bm{k}}} \,  2\omega_{k}\, p(\bm{k})^2
g(\omega + \omega_{\bm{k}}) ^2  \,.
\end{align*}

\end{proof}
\subsection{Proof of Proposition \ref{lemm3}}\label{proflemm3}
\begin{proof}  
Using Lemma  \ref{lemm1} and Lemma \ref{lemm2} we write the operator $W^{\pm}_{\Theta}$ as 
\begin{align*}
  W_{\Theta}^{\pm}&= \frac{1}{2}\iint \dd[3]{\bm{k}'}\dd[3]{\bm{k}} \,  p(\bm{k})p(\bm{k}')\\
	& \nonumber
\quad \Bigl(\Bigl(e^{-\frac{\theta}{4}\Vert k'-k \Vert^2}\, \int_{0}^{\infty} \dd{\omega}\, g(\omega-\omega_{\bm{k}'}) g(\omega-\omega_{\bm{k}})\Bigr)\, e^{ik\Theta  k'}
  a^*_{\Theta} (\bm{k}' )a_{\Theta} (\bm{k}  )  \\&\nonumber
  \quad +\Bigl( e^{-\frac{\theta}{4}\Vert k'-k \Vert^2}\int_{0}^{\infty} \dd{\omega}\,  g(\omega+\omega_{\bm{k}'})  g(\omega+\omega_{\bm{k}})\Bigr)\, 
 a^*_{\Theta}(\bm{k}) \,a_{\Theta}(\bm{k}' )  
\\& \nonumber \quad \pm \Bigl(e^{-\frac{\theta}{4}\Vert k'+k \Vert^2}\,\int_{0}^{\infty} \dd{\omega}\,  g(\omega+\omega_{\bm{k}'})\, g(\omega-\omega_{\bm{k}})  \Bigr)\, a_{\Theta}(\bm{k} ) a_{\Theta} (\bm{k}' ) \\& \quad \pm\Bigl(e^{-\frac{\theta}{4}\Vert k'+k \Vert^2}\,\int_{0}^{\infty} \dd{\omega}\, g(\omega-\omega_{\bm{k}'})  g(\omega+\omega_{\bm{k}}) \Bigr)a^*_{\Theta} (\bm{k} )a_{\Theta}^*(\bm{k}'  )  
  \nonumber \\
	& \nonumber
\quad + \Bigl(e^{-\frac{\theta}{4}\Vert k'-k \Vert^2}\, \int_{0}^{\infty} \dd{\omega}\, g(\omega+\omega_{\bm{k}}) g(\omega+\omega_{\bm{k}'})\Bigr)\, e^{ik\Theta  k'}
  a^*_{\Theta} (\bm{k}' )a_{\Theta} (\bm{k}  )  \\&\nonumber
  \quad +\Bigl( e^{-\frac{\theta}{4}\Vert k'-k \Vert^2}\int_{0}^{\infty} \dd{\omega}\,  g(\omega-\omega_{\bm{k}})  g(\omega-\omega_{\bm{k}'})\Bigr)\, 
 a^*_{\Theta}(\bm{k}) a_{\Theta}(\bm{k}' )  
\\& \nonumber \quad \pm \Bigl(e^{-\frac{\theta}{4}\Vert k'+k \Vert^2}\,\int_{0}^{\infty} \dd{\omega}\,  g(\omega+\omega_{\bm{k}})\, g(\omega-\omega_{\bm{k}'})  \Bigr)\, a_{\Theta}(\bm{k} ) a_{\Theta} (\bm{k}' ) \\& \quad \pm\Bigl(e^{-\frac{\theta}{4}\Vert k'+k \Vert^2}\,\int_{0}^{\infty} \dd{\omega}\, g(\omega-\omega_{\bm{k}})  g(\omega+\omega_{\bm{k}'}) \Bigr)\, a^*_{\Theta} (\bm{k} )a_{\Theta}^*(\bm{k}'  )  
  \Bigr) \nonumber
\\& \quad +
\int_0^\infty \dd{\omega} \int {\dd[3]{\bm{k}}}\,  2\omega_{k}\, p(\bm{k})^2
g(\omega + \omega_{\bm{k}}) ^2  \,,
\end{align*}
and summarizing the former terms renders 
\begin{align*}   
  W_{\Theta}^{\pm}&=\frac{1}{2}\iint \dd[3]{\bm{k}}'\dd[3]{\bm{k}} \,  p(\bm{k})p(\bm{k}') \Bigl(
F_{\theta}(\bm{k},\bm{k}')\, \Bigl( e^{ik\Theta  k'}
  a^*_{\Theta} (\bm{k}' )a_{\Theta} (\bm{k}  ) +
 a^*_{\Theta}(\bm{k}) \,a_{\Theta}(\bm{k}' )  \Bigr)
  \\& \quad \pm  {G}_{\theta}(\bm{k},\bm{k}')\, \Bigl(a_{\Theta}(\bm{k} ) a_{\Theta} (\bm{k}' )
    +a^*_{\Theta} (\bm{k} )a_{\Theta}^*(\bm{k}'  )\Bigr)  \Bigr)  +
\int_0^\infty \dd{\omega} \int  {\dd[3]{\bm{k}}} \,  2\omega_{k}\, p(\bm{k})^2
g(\omega + \omega_{\bm{k}}) ^2  \,  .
\end{align*}
\end{proof}  

\bibliographystyle{utphys}
\bibliography{refs.bib}

\providecommand{\href}[2]{#2}\begingroup\raggedright\begin{thebibliography}{10}

\bibitem{waldmannpos}
D.~Kaschek, N.~Neumaier, and S.~Waldmann, ``{Complete positivity of Rieffel’s
  deformation quantization by actions of $\mathbb{R}^d$},''
  \href{https://dx.doi.org/10.4171/JNCG/40}{{\em J. Noncommut. Geom} {\bfseries
  3} no.~3, (2009) 361--375}.

\bibitem{Few98}
C.~J. Fewster and S.~P. Eveson, ``Bounds on negative energy densities in flat
  spacetime,'' \href{https://dx.doi.org/10.1103/physrevd.58.084010}{{\em
  Physical Review D} {\bfseries 58} no.~8, (Sept., 1998) 084010}.

\bibitem{IN8}
C.~J. Fewster, ``Lectures on quantum energy inequalities,''
  \href{https://arxiv.org/abs/1208.5399}{{\ttfamily arXiv:1208.5399 [gr-qc]}}.

\bibitem{Kontou_2020}
E.-A. Kontou and K.~Sanders, ``Energy conditions in general relativity and
  quantum field theory,''
  \href{https://dx.doi.org/10.1088/1361-6382/ab8fcf}{{\em Classical and Quantum
  Gravity} {\bfseries 37} no.~19, (Sep, 2020) 193001}.

\bibitem{NCQEI1}
L.~H. Ford, ``Constraints on negative-energy fluxes,''
  \href{https://dx.doi.org/10.1103/PhysRevD.43.3972}{{\em Phys. Rev. D}
  {\bfseries 43} (Jun, 1991) 3972--3978}.

\bibitem{NCQEI2}
A.~D. Helfer, ``The stress - energy operator,''
  \href{https://dx.doi.org/10.1088/0264-9381/13/11/002}{{\em Classical and
  Quantum Gravity} {\bfseries 13} (1996) L129}.

\bibitem{NCQEI3}
L.~H. Ford and T.~A. Roman, ``Restrictions on negative energy density in flat
  spacetime,'' \href{https://dx.doi.org/10.1103/physrevd.55.2082}{{\em Physical
  Review D} {\bfseries 55} no.~4, (Feb., 1997) 2082–2089}.

\bibitem{NCQEI4}
M.~J. Pfenning and L.~H. Ford, ``Scalar field quantum inequalities in static
  spacetimes,'' \href{https://dx.doi.org/10.1103/PhysRevD.57.3489}{{\em Phys.
  Rev. D} {\bfseries 57} (Mar, 1998) 3489--3502}.

\bibitem{NCQEI5}
M.~J. Pfenning and L.~H. Ford, ``The unphysical nature of `warp drive’,''
  \href{https://dx.doi.org/10.1088/0264-9381/14/7/011}{{\em Classical and
  Quantum Gravity} {\bfseries 14} no.~7, (July, 1997) 1743–1751}.

\bibitem{Fewster_1999}
C.~J. Fewster and E.~Teo, ``Bounds on negative energy densities in static
  space-times,'' \href{https://dx.doi.org/10.1103/physrevd.59.104016}{{\em
  Physical Review D} {\bfseries 59} no.~10, (Apr., 1999) 104016}.

\bibitem{Fewster_2000}
C.~J. Fewster, ``A general worldline quantum inequality,''
  \href{https://dx.doi.org/10.1088/0264-9381/17/9/302}{{\em Classical and
  Quantum Gravity} {\bfseries 17} no.~9, (Apr., 2000) 1897–1911}.

\bibitem{Fewster_2002}
C.~J. Fewster and R.~Verch, ``A quantum weak energy inequality for dirac fields
  in curved spacetime,'' \href{https://dx.doi.org/10.1007/s002200100584}{{\em
  Communications in Mathematical Physics} {\bfseries 225} no.~2, (Feb., 2002)
  331–359}.

\bibitem{Fewster_2003}
C.~J. Fewster and B.~Mistry, ``Quantum weak energy inequalities for the dirac
  field in flat spacetime,''
  \href{https://dx.doi.org/10.1103/physrevd.68.105010}{{\em Physical Review D}
  {\bfseries 68} no.~10, (Nov., 2003) 105010}.

\bibitem{Fewster_2003a}
C.~J. Fewster and M.~J. Pfenning, ``A quantum weak energy inequality for
  spin-one fields in curved space–time,''
  \href{https://dx.doi.org/10.1063/1.1602554}{{\em Journal of Mathematical
  Physics} {\bfseries 44} no.~10, (Oct., 2003) 4480–4513}.

\bibitem{Fewster_2005}
C.~Fewster and S.~Hollands, ``Quantum energy inequalities in two-dimensional
  conformal field theory,''
  \href{https://dx.doi.org/10.1142/s0129055x05002406}{{\em Reviews in
  Mathematical Physics} {\bfseries 17} no.~05, (June, 2005) 577–612}.

\bibitem{Fewster_2006}
S.~P. Dawson and C.~J. Fewster, ``An explicit quantum weak energy inequality
  for dirac fields in curved spacetimes,''
  \href{https://dx.doi.org/10.1088/0264-9381/23/23/005}{{\em Classical and
  Quantum Gravity} {\bfseries 23} no.~23, (Oct., 2006) 6659–6681}.

\bibitem{Fewster_2006a}
C.~J. Fewster and M.~J. Pfenning, ``Quantum energy inequalities and local
  covariance. i. globally hyperbolic spacetimes,''
  \href{https://dx.doi.org/10.1063/1.2212669}{{\em Journal of Mathematical
  Physics} {\bfseries 47} no.~8, (Aug., 2006) 105010}.

\bibitem{Fewster_2006b}
C.~J. Fewster, ``Quantum energy inequalities and local covariance ii:
  categorical formulation,''
  \href{https://dx.doi.org/10.1007/s10714-007-0494-3}{{\em General Relativity
  and Gravitation} {\bfseries 39} no.~11, (Sept., 2007) 1855–1890}.

\bibitem{Fewster_2007}
S.~P. Eveson and C.~J. Fewster, ``Mass dependence of quantum energy inequality
  bounds,'' \href{https://dx.doi.org/10.1063/1.2779137}{{\em Journal of
  Mathematical Physics} {\bfseries 48} no.~9, (Sept., 2007) 093506}.

\bibitem{Fewster_2007b}
C.~J. Fewster and L.~W. Osterbrink, ``Quantum energy inequalities for the
  non-minimally coupled scalar field,''
  \href{https://dx.doi.org/10.1088/1751-8113/41/2/025402}{{\em Journal of
  Physics A: Mathematical and Theoretical} {\bfseries 41} no.~2, (Dec., 2007)
  025402}.

\bibitem{Fewster_2008}
C.~J. Fewster and C.~J. Smith, ``Absolute quantum energy inequalities in curved
  spacetime,'' \href{https://dx.doi.org/10.1007/s00023-008-0361-0}{{\em Annales
  Henri Poincaré} {\bfseries 9} no.~3, (May, 2008) 425–455}.

\bibitem{Fewster_2009}
H.~Bostelmann and C.~J. Fewster, ``Quantum inequalities from operator product
  expansions,'' \href{https://dx.doi.org/10.1007/s00220-009-0853-x}{{\em
  Communications in Mathematical Physics} {\bfseries 292} no.~3, (July, 2009)
  761–795}.

\bibitem{Boss}
H.~Bostelmann, D.~Cadamuro, and C.~J. Fewster, ``Quantum energy inequality for
  the massive ising model,''
  \href{https://dx.doi.org/10.1103/physrevd.88.025019}{{\em Physical Review D}
  {\bfseries 88} no.~2, (July, 2013) 025019}.

\bibitem{fröb2023quantumenergyinequalitysinegordon}
M.~B. Fr\"ob and D.~Cadamuro, ``{A quantum energy inequality in the
  Sine--Gordon model},'' \href{https://arxiv.org/abs/2212.07377}{{\ttfamily
  arXiv:2212.07377 [math-ph]}}.

\bibitem{DFR}
S.~Doplicher, K.~Fredenhagen, and J.~E. Roberts, ``{The Quantum structure of
  space-time at the Planck scale and quantum fields},''
  \href{https://dx.doi.org/10.1007/BF02104515}{{\em Commun. Math. Phys.}
  {\bfseries 172} (1995) 187--220}.

\bibitem{GL1}
H.~Grosse and G.~Lechner, ``{Wedge-Local Quantum Fields and Noncommutative
  Minkowski Space},''
  \href{https://dx.doi.org/10.1088/1126-6708/2007/11/012}{{\em JHEP} {\bfseries
  0711} (2007) 012}.

\bibitem{BLS}
D.~Buchholz, G.~Lechner, and S.~J. Summers, ``{Warped Convolutions, Rieffel
  Deformations and the Construction of Quantum Field Theories},''
  \href{https://dx.doi.org/10.1007/s00220-010-1137-1}{{\em Commun. Math. Phys.}
  {\bfseries 304} (2011) 95--123}.

\bibitem{NCEMT2000}
A.~Gerhold, J.~Grimstrup, H.~Grosse, L.~Popp, M.~Schweda, and R.~Wulkenhaar,
  ``{The Energy momentum tensor on noncommutative spaces. Some pedagogical
  comments},'' {\em Ukr. J. Phys.} {\bfseries 47} (2002) 219--225.

\bibitem{NCEMT2001}
M.~Abou-Zeid and H.~Dorn, ``Comments on the energy–momentum tensor in
  non-commutative field theories,''
  \href{https://dx.doi.org/10.1016/s0370-2693(01)00780-8}{{\em Physics Letters
  B} {\bfseries 514} no.~1–2, (Aug., 2001) 183–188}.

\bibitem{NCEMT2003}
A.~K. Das and J.~Frenkel, ``{On the energy momentum tensor in noncommutative
  gauge theories},'' \href{https://dx.doi.org/10.1103/PhysRevD.67.067701}{{\em
  Phys. Rev. D} {\bfseries 67} (2003) 067701}.

\bibitem{NCEMT2015}
H.~Balasin, D.~N. Blaschke, F.~Gieres, and M.~Schweda, ``{On the
  energy-momentum tensor in Moyal space},''
  \href{https://dx.doi.org/10.1140/epjc/s10052-015-3492-8}{{\em Eur. Phys. J.
  C} {\bfseries 75} no.~6, (2015) 284}.

\bibitem{NCEMT2018}
E.~Baloitcha, V.~Lahoche, and D.~O. Samary, ``Energy momentum tensor for
  translation invariant renormalizable noncommutative field theory,''
  \href{https://dx.doi.org/10.1140/epjp/i2018-12339-8}{{\em The European
  Physical Journal Plus} {\bfseries 133} no.~12, (Dec., 2018) 515}.

\bibitem{GL2}
H.~Grosse and G.~Lechner, ``{Noncommutative Deformations of Wightman Quantum
  Field Theories},''
  \href{https://dx.doi.org/10.1088/1126-6708/2008/09/131}{{\em JHEP} {\bfseries
  0809} (2008) 131}, \href{https://arxiv.org/abs/0808.3459}{{\ttfamily
  arXiv:0808.3459 [math-ph]}}.

\bibitem{Gauss2}
P.~Nicolini, ``Noncommutative black holes, the final appeal to quantum gravity:
  A review,'' \href{https://dx.doi.org/10.1142/s0217751x09043353}{{\em
  International Journal of Modern Physics A} {\bfseries 24} no.~07, (Mar.,
  2009) 1229–1308}.

\bibitem{Gauss3}
S.~Kobayashi, ``Regular black holes and noncommutative geometry inspired fuzzy
  sources,'' \href{https://dx.doi.org/10.1142/s0217751x16500809}{{\em
  International Journal of Modern Physics A} {\bfseries 31} no.~14n15, (May,
  2016) 1650080}.

\bibitem{Gauss4}
M.-S. Ma and R.~Zhao, ``Noncommutative geometry inspired black holes in rastall
  gravity,'' \href{https://dx.doi.org/10.1140/epjc/s10052-017-5217-7}{{\em The
  European Physical Journal C} {\bfseries 77} no.~9, (Sept., 2017) 629}.

\bibitem{Gauss5}
Z.~Yan, C.~Wu, and W.~Guo, ``Scalar field quasinormal modes of noncommutative
  high dimensional schwarzschild-tangherlini black hole spacetime with smeared
  matter sources,''
  \href{https://dx.doi.org/https://doi.org/10.1016/j.nuclphysb.2020.115217}{{\em
  Nuclear Physics B} {\bfseries 961} (2020) 115217}.

\bibitem{Gauss1}
M.~Pinkwart-Walker, {\em Quantum Aspects of Cosmology}.
\newblock PhD thesis, 2020.
\newblock \url{http://urn-resolving.de/urn:nbn:de:gbv:579-opus-1009157}.

\bibitem{Fewster_2003b}
C.~J. Fewster and R.~Verch, ``{Stability of Quantum Systems at Three Scales:
  Passivity, Quantum Weak Energy Inequalities and the Microlocal Spectrum
  Condition},'' \href{https://dx.doi.org/10.1007/s00220-003-0884-7}{{\em
  Commun. Math. Phys.} {\bfseries 240} (2003) 329--375}.

\bibitem{Fewster_2007a}
C.~J. Fewster, {\em Quantum Energy Inequalities and Stability Conditions in
  Quantum Field Theory},
  \href{https://dx.doi.org/10.1007/978-3-7643-7434-1_8}{pp.~95--111}.
\newblock Birkh{\"a}user Basel, Basel, 2007.

\bibitem{RAD2}
M.~J. Radzikowski, ``{Micro-local approach to the Hadamard condition in quantum
  field theory on curved space-time},''
  \href{https://dx.doi.org/10.1007/BF02100096}{{\em Commun. Math. Phys.}
  {\bfseries 179} (1996) 529--553}.

\bibitem{SV01}
H.~Sahlmann and R.~Verch, ``Microlocal spectrum condition and hadamard form for
  vector-valued quantum fields in curved spacetime,''
  \href{https://dx.doi.org/10.1142/s0129055x01001010}{{\em Reviews in
  Mathematical Physics} {\bfseries 13} no.~10, (Oct, 2001) 1203–1246}.

\bibitem{IN3}
R.~Brunetti, K.~Fredenhagen, and M.~Köhler, ``The microlocal spectrum
  condition and wick polynomials of free fields on curved spacetimes,''
  \href{https://dx.doi.org/10.1007/bf02099626}{{\em Communications in
  Mathematical Physics} {\bfseries 180} no.~3, (Oct, 1996) 633–652}.

\bibitem{SVW}
A.~{Strohmaier}, R.~{Verch}, and M.~{Wollenberg}, ``{Microlocal Analysis of
  Quantum Fields on Curved Space-times: Analytic Wave Front Sets and
  Reeh-Schlieder Theorems},'' \href{https://dx.doi.org/10.1063/1.1506381}{{\em
  J. Math. Phys.} {\bfseries 43} (2002) 5514--5530}.

\bibitem{Sanders}
K.~{Sanders}, ``{Equivalence of the (Generalised) Hadamard and Microlocal
  Spectrum Condition for (Generalised) Free Fields in Curved Spacetime},''
  \href{https://dx.doi.org/10.1007/s00220-009-0900-7}{{\em Commun. Math. Phys.}
  {\bfseries 295} (2010) 485--501}.

\bibitem{BMV}
R.~Ballal, A.~Much, and R.~Verch, ``Microlocal analysis of a deformed quantum
  field theory,'' \href{https://arxiv.org/abs/2412.20511}{{\ttfamily
  arXiv:2412.20511 [math-ph]}}.

\bibitem{G}
H.~Grosse, ``{On the construction of M{\"o}ller operators for the nonlinear
  Schr{\"o}dinger equation},''
  \href{https://dx.doi.org/10.1016/0370-2693(79)90835-9}{{\em Phys. Lett.}
  {\bfseries B 86} (1979) 267--271}.

\end{thebibliography}\endgroup

\end{document}